\documentclass[reqno,11pt]{amsart}
\usepackage{amsmath,amsthm,amscd,amssymb,latexsym,upref,stmaryrd}
\usepackage{amsfonts,mathrsfs,dsfont,bm}
\usepackage[top=1.5in, bottom=1.5in, left=1.4in, right=1.4in]{geometry}
\usepackage{color}

\newcommand{\IR}{\mathbb{R}}
\newcommand{\ID}{\mathbb{D}}
\newcommand{\IC}{\mathbb{C}}
\newcommand{\IV}{\mathbb{V}}

\newcommand{\bsA}{{\boldsymbol{A}}}
\newcommand{\bsB}{{\boldsymbol{B}}}
\newcommand{\bsC}{{\boldsymbol{C}}}
\newcommand{\bsD}{{\boldsymbol{D}}}
\newcommand{\bsE}{{\boldsymbol{E}}}

\newcommand{\bsG}{{\boldsymbol{G}}}
\newcommand{\bsH}{{\boldsymbol{H}}}

\newcommand{\bsM}{{\boldsymbol{M}}}
\newcommand{\bsN}{{\boldsymbol{N}}}
\newcommand{\bsP}{{\boldsymbol{P}}}

\newcommand{\bsV}{{\boldsymbol{V}}}

\newcommand{\bsPsi}{{\boldsymbol\Psi}}

\newcommand{\bsPhin}{{\boldsymbol{\Phi_n}}}
\newcommand{\bsPhi}{{\boldsymbol\Phi}}
\newcommand{\bsXi}{{\boldsymbol\Xi}}
\newcommand{\bscE}{\boldsymbol{\mathcal{E}}}
\newcommand{\bscH}{\boldsymbol{\mathcal{H}}}
\newcommand{\bsvareps}{{\boldsymbol\varepsilon}}
\newcommand{\bsmu}{{\boldsymbol\mu}}
\newcommand{\bsAj}{{\boldsymbol{A^j}}}
\newcommand{\bsAl}{{\boldsymbol{A^l}}}
\newcommand{\bsaj}{{\boldsymbol{a^j}}}

\newcommand{\bsal}{{\boldsymbol{a^l}}}

\newcommand{\bsxi}{{\boldsymbol{\xi}}}
\newcommand{\bssig}{{\boldsymbol{\sigma}}}

\newcommand{\cE}{{\mathcal E}}

\newcommand{\cH}{{\mathcal H}}

\newcommand{\dlab}{\langle\!\langle}
\newcommand{\drab}{\rangle\!\rangle}

\newcommand{\nMh}{\nabla_{\!\!\hatt{\mathcal M}}}

\newcommand{\cMh}{\hatt{\mathcal M}}
\newcommand{\Mh}{\hatt{\bsM}}
\newcommand{\nM}{\nabla_{\!\mathcal M}}

\DeclareMathOperator{\supp}{supp}

\newcommand{\beq}{\begin{equation}}
\newcommand{\enq}{\end{equation}}

\newcommand{\Om}{\Omega}

\newcommand{\wti}{\widetilde}

\newcommand{\hatt}{\widehat}

\let\geq\geqslant
\let\leq\leqslant

\makeatletter
\def\theequation{\@arabic\c@equation}

\allowdisplaybreaks 
\numberwithin{equation}{section}

\newtheorem{theorem}{Theorem}[section]

\newtheorem{lemma}[theorem]{Lemma}
\newtheorem{corollary}[theorem]{Corollary}

\theoremstyle{remark}

\newtheorem{example}[theorem]{Example}

\begin{document}

\title[First order diff. eqns. on $\IR^d$]{On essential self-adjointness for first order differential operators on domains in $\IR^d$}

\author[G.\ Nenciu]{Gheorghe Nenciu}
\address{Gheorghe Nenciu\\Institute of Mathematics ``Simion Stoilow'' of the Romanian Academy\\ 21, Calea Grivi\c tei\\010702-Bucharest, Sector 1\\Romania}
\email{Gheorghe.Nenciu@imar.ro}

\author[I.\ Nenciu]{Irina Nenciu}
\address{Irina Nenciu\\
         Department of Mathematics, Statistics and Computer Science\\ 
         University of Illinois at Chicago\\
         851 S. Morgan Street\\
         Chicago, IL \textit{and} Institute of Mathematics ``Simion Stoilow''
     of the Romanian Academy\\ 21, Calea Grivi\c tei\\010702-Bucharest, Sector 1\\Romania}
\email{nenciu@uic.edu}

\thanks{The research of I.N. is partly supported by NSF
grant DMS-1150427 and by a Simons Foundation Fellowship in mathematics.}

\begin{abstract}
We consider general symmetric systems of first order linear partial differential operators on
domains $\Omega \subset \IR^d$, and we seek sufficient conditions on the 
coefficients which ensure essential self-adjointness. The coefficients of the first order terms 
are only required to belong to $C^1(\Omega)$ and there is no ellipticity 
condition. Our criterion writes as the completeness of an 
associated Riemannian structure which encodes the propagation velocities of the system.
As an application we obtain sufficient conditions for confinement of energy for some wave propagation 
problems of classical physics.
\end{abstract}

\maketitle


\section{Introduction}\label{S:1}

In this note, we consider the essential self-adjointness problem of formally
symmetric first order differential operators on domains $\Om\subset\IR^d$. More precisely, 
consider, on $\Om\subset\IR^d$, the first order differential operator
\begin{equation}\label{E:1.1}
\ID=\bsE^{-1}\bigg(\sum_{j=1}^d \frac12\big(\bsAj D_j+D_j\bsAj\big)+\bsV\bigg)\,
\end{equation}
where $\bsE,\bsAj,\bsV$ are $k\times k$ matrix-valued functions and 
$D_j=-i\frac{\partial}{\partial x_j}$.
Assuming that $\bsE(x)>0$, $\bsAj(x)=\big(\bsAj(x)\big)^*$, $\bsV(x)=\big(\bsV(x)\big)^*$ for all $x\in\Om$,
$\bsE,\bsAj\in C^1(\Om)^{k\times k}$, and $\bsV\in L^2_\text{loc}(\Om)^{k\times k}$,
we obtain that
$\ID$ is symmetric on $C_0^1(\Om)^{k}$ with respect to the ``energy" scalar product
\begin{equation*}
\dlab\bsPhi,\bsPsi\drab_\bsE=\sum_{j=1}^k \int_\Om \overline{\Phi_j(x)}
\big(\bsE(x)\bsPsi\big)_j(x)\,dx\,.
\end{equation*}
The problem is to find sufficient conditions on coefficients $\bsE,\bsAj,\bsV$ which ensure
that $\ID$ is essentially self-adjoint.

 Due to the fact that, by Stone's theorem,
 self-adjointness is equivalent with the existence of unique, norm-preserving
global-in-time solutions of the corresponding evolution equation, 
essential self-adjointness of  symmetric differential operators is an old
and fundamental problem of mathematical physics 
 which naturally has a long and ramified history. This history can be traced from 
\cite{BMS,Br,KSWW,RS}  for second order
partial differential operators, and from \cite{Ch1,Ch2,Th,Wi1} for
first order partial differential operators, not only on domains in $\IR^d$
but also in a general Riemannian setting (see also \cite{NN2,MT, LM, PRS, BFSB, FR, EL} and references therein for some recent developments).
In particular, in a well-known
paper \cite{Ch1} Chernoff proved that all the powers of formally symmetric first
order differential operators with $C^{\infty}$ coefficients on complete Riemannian manifolds
are essential self-adjoint provided the ``velocity of propagation" does not increase
too fast at ``infinity".  In a companion paper \cite{Ch2}, he used the results in \cite{Ch1} to prove essential 
self-adjointness for large classes of Dirac and Schr\"odinger operators with singular potentials.
The main idea of his method is to employ standard results of Friedrichs'
theory of hyperbolic equations to construct local-in-time solutions, and then
to use the method of local energy inequalities \cite{Wi2,Ch1} to
prove the existence, uniqueness, and smoothness of global-in-time solutions.
This in turn gives essential self-adjointness. The heuristic argument behind the proof is that,
under the condition that the velocity of propagation does not increase too fast at infinity, a compactly supported
initial condition remains compactly supported for all $t>0$, 
 and hence the evolution is completely determined by the coefficients of $\mathbb D$. 

Particularizing Chernoff's result to the flat case (i.e. the Riemannian manifold
is just $\IR^d$ with the standard metric), assuming  that $\bsE(x)={\bf{1}}$ and   $\bsAj$, 
$\bsV$  are $C^{\infty}$, one obtains that all the powers of $\mathbb D$ are essential self-adjoint provided
\begin{equation}\label{Ch1}
\int^\infty\frac{1}{b(r)}\,dr=\infty
\end{equation} 
where
\begin{equation}\label{Ch2}
b(r)=\sup_{|x|\leq r} \sup_{\xi\in\IR^d,|\xi|=1} 
\bigg\|\sum_{j=1}^d \xi_j\bsAj(x)\bigg\|_{\mathcal B(\IC^k)}\,.
\end{equation}
Notice that, beside the fact that one obtains essential self-adjointness for all powers of $\mathbb D$ 
and there are no conditions on the behavior of $\bsV$ at infinity, no ellipticity is required. 
Still the result is not completely satisfactory in some respects.  If we are interested only in the essential 
self-adjointness of $\mathbb D$ one expects that much less regularity
of the coefficients is needed and while there are no problems with  lowering the regularity of $\bsV$  \cite{Ch2}, 
Chernoff's method still requires $\bsAj$ to be  $C^{\infty}$. Further,   the  ``isotropic" condition  \eqref{Ch2} 
involves the supremum of $\|\bsAj\|$ on balls and seems stronger that what the heuristics suggests. Finally, the 
method is suited only for ${\mathbb R}^d$ and not for arbitrary domains $\Omega\subset\IR^d$.

The aim of this note is to point out that by using the  classical result of
Friedrichs \cite{Fr} giving the equality of ``weak" and ``strong"
extensions of  first order differential operators one can 
generalize the proof for standard Dirac operator (see e.g. Theorem 4.3 in \cite{Th}) to obtain a satisfactory criterion of essential self-adjointness for $\mathbb D$ on arbitrary domains $\Omega \subset \IR^d$. Our criterion is of the same
type as the one given in \cite{BMS} for essential self-adjointness of  operators of the form
$D^*D$,  where $D$ is a first order elliptic differential operator with smooth coeficients,  namely the completeness of an associated Riemannian structure associated to $D$. More precisely, let $\bsM(x)$  be the real, non-negative definite matrix given by:
\begin{equation*}
\bsM_{jl}(x)=\text{Tr}\big(\bsE(x)^{-1/2}\bsAj(x)\bsE(x)^{-1}\bsAl(x)\bsE(x)^{-1/2}\big)\,.
\end{equation*}
Then our first  result (see Theorem~\ref{T:1}i. below) states that if
$\bsV\in L^\infty_\text{loc}(\Om)^{k\times k}$, and if for some $\hatt\bsM\geq\bsM$,
$0<\hatt\bsM\in C^\infty(\Om)^{k\times k}$, $\Om$ endowed with
the Riemannian metric given by
\begin{equation*}
ds^2=\sum_{j,l=1}^d {\hatt\bsM}^{-1}(x)_{jl}\,dx_jdx_l
\end{equation*}
is complete, then $\ID$ is essentially self-adjoint. Our second result, given in Theorem~\ref{T:1}ii. is the analog in our setting of Theorem 2.1 in \cite{Ch2}.

After our note was ready for submission, we found an old, unnoticed paper by Fattorini \cite{Fa} proving a result implying the essential self-adjointness of $\mathbb D$ for the case  $\Omega = {\mathbb R}^d$,  
$\bsE(x)={\bf{1}}$ ,  $\bsAj$ 
are $C^1$ and $\bsV$ is continuous. His condition about behavior at infinity of $\bsAj$ is similar to 
the Chernoff condition:
\begin{equation}\label{AF1}
\|\bsAj(x)\| \leq \rho(|x|) , 
\end{equation}
where $\rho$ is a positive continuous function growing so slowly to infinity that
\begin{equation}\label{AF2}
\int^\infty\frac{1}{\rho(r)}\,dr=\infty.
\end{equation}
His proof also uses as the essential ingredient the Friedrichs result about identity between strong and weak extensions \cite{Fr, Ho}.

The matrix $\bsM(x)$  encodes the propagation velocities
of the system (see Section ~\ref{S:4} for examples) and may be called the "velocity matrix" associated to $\ID$.
Our condition of completeness of $(\Omega, {\hatt\bsM}^{-1})$ can be viewed as a weaker "anisotropic" form
of the Chernoff and Fattorini conditions (see the end of Section~\ref{S:3} for details).

The content of the paper is as follows. In Section~\ref{S:2} we state and prove our
main result. In Section~\ref{S:3} we give some corollaries and in Section~\ref{S:4}
some applications of Theorem~\ref{T:1}. In particular, we obtain criteria for essential self-adjointness
of first order differential operators describing some wave propagation
phenomena of classical physics \cite{Wi1}. For example in the case of propagation
of acoustical waves in an isotropic inhomogeneous medium at rest, the matrix
$\bsM(x)$ turns out to be
\begin{equation*}
\bsM(x)=2\big(v_P^2(x)+2v_S^2(x)\big){\bf{1}},\quad\text{for }x\in\Omega\,,
\end{equation*}
where $v_P(x)$, $v_S(x)$ are the local  compressional and  shear velocities respectively. Hence, our result is 
indeed the mathematical substantiation of the heuristic argument that the confinement
is nothing but the fact that an initial disturbance inside $\Om$ will never reach the boundary.

We conclude with two remarks. The first one is about optimality. Since the growth condition \eqref{AF2} is optimal in some sense (see the counterexample in \cite{Fa})  it is tempting
to think that the completeness of the associated Riemannian structure
on $\Om$ is also necessary for the essential self-adjointness of $\ID$. However 
this (like for the second order differential operators case)  is not the case. For example, 
the standard free Dirac operator on $\IR^3\setminus\{0\}$ is essentially
self-adjoint on $C_0^1(\IR^3\setminus\{0\})^4$, but in this case
$\bsM(x)=4\cdot{\bf{1}}$. The second remark is that, while we
restricted ourselves to the flat case $\Om\subset\IR^d$, by adding
the necessary technicalities the method of this paper can be used to lift
the results to a more general Riemannian setting.

\section{The main result}\label{S:2}

We consider symmetric first order differential operators on domains in $\Om\subset\IR^d$ 
of the form \eqref{E:1.1} which is a short hand for 
\begin{equation}\label{E:2.4}
\begin{aligned}
(\mathbb D\bsPsi)_l(x)
&=-\frac{i}{2}\sum_{m,n=1}^k\sum_{j=1}^d \bsE^{-1}_{lm}(x)\bigg[\boldsymbol{A^j}_{mn}(x)\,\frac{\partial\bsPsi_n(x)}{\partial x_j} 
  +\frac{\partial}{\partial x_j}\bigg(\boldsymbol{A^j}_{mn}(x)\bsPsi_n(x)\bigg)\bigg]\\
  &\qquad+\sum_{m,n=1}^k \bsE^{-1}_{lm}(x)\bsV_{mn}(x)\bsPsi_n(x)\,.
\end{aligned}
\end{equation}
For a $p\times q$ matrix-valued function $\bsB(x)$, we will abuse notation slightly
and write, for example, $\bsB\in C_0^1(\Om)$ if all of its matrix entries are in $C_0^1(\Omega)$, etc. 
For the rest of the paper, $\bsE(x)$, $\boldsymbol{A^j}(x)$, $j=1,...d$, and $\bsV(x)$ denote $k\times k$ Hermitian matrices,
$\bm{\Phi}(x)$ and $\bm{\Psi}(x)$ denote $k\times 1$ matrices, and $D_j=-i\frac{\partial}{\partial x_j}$ for $j=1,...,d$.

We make the following assumptions:
\begin{equation}\label{E:2.1}
\bsE,\boldsymbol{A^j}\in C^1(\Om)\,;\quad \bsE>0\,;\quad \bsV\in L^2_{loc}(\Om)\,.
\end{equation}
We denote by $\mathcal H_\bsE$ the space obtained by the completion of $\big(C_0^1(\Omega)\big)^k$ in the norm
given by the scalar product
\begin{equation}\label{E:2.5}
\dlab \bsPhi,\bsPsi\drab_{_\bsE}=\int_\Om \langle \bsPhi(x),\bsE(x)\bsPsi(x)\rangle\,dx
\end{equation}
where $\langle\cdot,\cdot\rangle$ is  the standard scalar product in $\mathbb C^k$.
In the particular case when $\bsE= \bf 1$ i.e.  $\mathcal H_\bsE=L^2(\Om)^k$  we omit the subscript and
write $\dlab \cdot, \cdot \drab$ and $\| \cdot\|$ for the scalar product and norm respectively.
For  $\bsN$ a real, positive-definite, $d\times d$ matrix-valued function with $\bsN\in C^\infty(\Omega)$,
we denote by $\mathcal N:=(\Omega, \bsN^{-1})$ the Riemannian structure on $\Omega$ given by
\begin{equation}\label{E:2.7}
ds^2=\sum_{j,l=1}^d \bsN^{-1}_{jl}(x)\,dx_jdx_l\,.
\end{equation} 
In this setting, our main result reads:

\begin{theorem}\label{T:1}
Consider the $k\times k$ matrix-valued functions
$\bsE,\boldsymbol{A^j}\in C^1(\Om)$, $j=1,...,d$, $\bsE>0$, and $\bsV\in L^2_{loc}(\Om)$
with which we define the symmetric operator in $\mathcal H_\bsE$:
\begin{equation}\label{E:2.3}
\mathbb D=\tfrac12 \bsE^{-1}\big(\bsA\cdot\bsD+\bsD\cdot\bsA\big)+\bsE^{-1}\bsV:= \mathbb D_0+\mathbb V\,,\quad\mathcal D(\ID)=\big\{\bsPsi\,\big|\, \bsPsi\in C_0^1(\Omega)\big\}\,.
\end{equation}
Let $\bsM$ be the $d\times d$ real non-negative  matrix-valued function on $\Om$ given by
\begin{equation}\label{E:2.6}
\bsM_{jl}={\rm{Tr}}\big(\bsE^{-1/2}\boldsymbol{A^j}\bsE^{-1}\boldsymbol{A^l}\bsE^{-1/2}\big)\,,\quad j,l=1,\dots,d\,.
\end{equation}
 Assume that there exists a real, positive-definite, $d\times d$ matrix-valued function $\hatt \bsM\in C^\infty(\Omega)$,
such that $\hatt{\mathcal M}$ is a complete Riemannian manifold and
\begin{equation}\label{E:2.7.1}
\bsM(x)\leq \Mh(x)\qquad\text{for all }x\in\Om\,.
\end{equation} 
Then:
\begin{itemize}
\item[i.]  
If $\bsV\in L^\infty_\text{loc}(\Om)$, then $\ID=\ID_0+\IV$ is essentially self-adjoint.
\item[ii.] Let $\bsV =\boldsymbol{V_1}+\boldsymbol{V_2}$, $\boldsymbol{V_1}, \boldsymbol{V_2}\in L^2_\text{loc}(\Om)$
and assume that for any compact set $K\subset\Om$, we have that $\ID_0+\IV_1+\chi_K\IV_2$ is
essentially self-adjoint. Here, as usual, $\chi_K$ is the characteristic function of the set $K$.
Then $\ID_0+\IV_1+\IV_2$ is essentially self-adjoint.
\end{itemize}
\end{theorem}

To prepare the proof of Theorem \ref{T:1}, we start by considering the case $\bsE\equiv \bf 1$ 
and $\bsV\equiv 0$.
In this simpler context, we will need a series of technical results, presented below.
Recall that in this simplified case, we have that
\begin{equation}\label{E:2.8}
\ID=\ID_0=\frac12\big(\bsA\cdot\bsD+\bsD\cdot\bsA\big)\quad\text{with } 
\mathcal D(\ID_0)=C_0^1(\Om)^k
\end{equation}
is a symmetric operator in $\mathcal H=L^2(\Om)^k$. We note here that it will sometimes be useful to write 
$\ID_0$ in its canonical form, namely
\begin{equation}\label{E:2.9}
\ID_0=\bsA\cdot\bsD-\frac{i}{2}\,\boldsymbol{\partial}\cdot\bsA
\end{equation}
where $\boldsymbol{\partial}\cdot\bsA$ is the operator of multiplication with
$\sum_{j=1}^d \frac{\partial \boldsymbol{A^j}(x)}{\partial x_j}$. In particular the symbol of $\ID_0$
is the Hermitian matrix
\begin{equation}\label{E:2.10}
\bssig(x,\xi)=\bsA(x)\cdot\xi\,,\qquad \xi\in\IR^d\,.
\end{equation}
For the case at hand, 
\begin{equation}\label{E:2.11}
\bsM_{jl}(x)=\text{Tr}\big(\boldsymbol{A^j}(x)\boldsymbol{A^l}(x)\big)\,.
\end{equation}

With the assumptions above, we have that $\bsM(x)$ is real and
for $\xi\in\IC^d$, we have
\begin{equation}\label{E:2.12}
\langle \xi, \bsM(x)\xi\rangle={\rm{Tr}}\big(\bssig(x,\xi)^*\bssig(x,\xi)\big)\,.
\end{equation}
In particular, we find from \eqref{E:2.12} that $\bsM(x)$ is non-negative definite for
all $x\in\Om$. Consider now a scalar-valued function $f$ on $\Om$. 
By a slight abuse of notation, we will also denote by $f$ the operator
of multiplication with $f\,{\bf{1}}_{k\times k}$. If $f\in C^1(\Om)$ with $f(x)=\overline{f(x)}$ for 
all $x\in\Om$, then on $C_0^1(\Om)$
\begin{equation}\label{E:2.13}
\big[f,\ID_0\big]=i\bssig\big(\cdot,\nabla f\big)\,.
\end{equation}

Let now, as in Theorem \ref{T:1}, $\hatt \bsM$ be smooth and such that $\hatt \bsM(x)>0$ and 
$\hatt \bsM(x)\geq \bsM(x)$. Recall that, in the Riemannian structure $\hatt{\mathcal M}=(\Omega, \hatt \bsM^{-1})$,
$\nMh f=\hatt \bsM \nabla f$ and the norm is
\begin{equation}\label{E:2.14}
\big|\nMh f\big|_{_{\hatt{\mathcal M}}}^2=\big\langle\nabla f,\hatt \bsM \nabla f\big\rangle\,.
\end{equation}
The following simple inequality is crucial to what follows.

\begin{lemma}\label{L:3}
Let $f=\bar f\in C^1(\Om)$ be a scalar-valued function. Then
\begin{equation}\label{E:2.15}
\big\|\bssig(x,\nabla f(x))\big\|^2_{\mathcal B(\IC^k)}\leq \big|\nMh f(x)\big|_{\hatt{\mathcal M}}^2
\qquad\text{for all }x\in\Om\,.
\end{equation}
\end{lemma}

\begin{proof}
Since $f$ is real-valued, it follows that $\bssig\big(x,\nabla f(x)\big)$ is Hermitian. Together with
\eqref{E:2.10}, \eqref{E:2.12}, and recalling the condition that $\hatt \bsM\geq \bsM$, this implies that
for any $x\in\Om$:
\begin{equation*}
\begin{aligned}
\big\|\bssig(x,\nabla f(x))\big\|^2_{\mathcal B(\IC^k)}&\leq \text{Tr}\, \bssig(x,\nabla f(x))^2=\langle \nabla f(x), \bsM(x)\nabla f(x)\rangle\\
&\leq \langle \nabla f(x), \hatt \bsM(x)\nabla f(x)\rangle=\big|\nMh f(x)\big|_{_{\hatt{\mathcal M}}}^2\,,
\end{aligned}
\end{equation*}
as claimed.
\end{proof}

The assumption that $\cMh$ is complete is used via the following lemma, giving the existence of
a Gaffney-type set of cut-off functions (see, e.g., Proposition 4.1 in \cite{Sh})

\begin{lemma}\label{L:4}
 There exists a sequence
of functions $G_p\,:\, \cMh\to[0,1]$, $p\geq 1$, such that:
\begin{itemize}
\item[i.] $G_p\in C_0^\infty(\cMh)$;
\item[ii.] For every compact $K\subset \cMh$ there exists $p_K\geq 1$ such that
\begin{equation*}
G_p\big|_{_K}=1\qquad\text{for all } p\geq p_K\,;
\end{equation*}
\item[iii.] $\lim_{p\to\infty}\sup_{x\in \cMh} \big|\nabla_{\!\!\cMh} G_p(x)\big|_{_{\cMh}}=0$.
\end{itemize}
\end{lemma}

As already discussed in the Introduction, our proof uses in an essential way the famous Friedrichs result giving the equality 
between strong and weak extensions of $\ID_0$. More precisely, let $\mathcal D(\ID_0^*)$ be the domain of the adjoint
$\ID_0^*$ of of $\ID_0$ i.e.
\begin{equation}\label{E:2.16}
\mathcal D(\ID_0^*)=
\big\{\bsPsi\,\big|\,\exists\, C_\bsPsi>0\text{ such that } |\dlab \ID_0\bsPhi,\bsPsi\drab|\leq C_\bsPsi\|\bsPhi\|\text{ for all }\bsPhi\in C_0^1(\Omega)\big\}
\end{equation}
and
\begin{equation}\label{E:2.17}
\begin{aligned}
\mathcal F=\big\{\bsPsi\,\big|\,&\exists\, \bsXi\in L^2(\Om)^k\text{ such that for every }\Om'\subset\!\subset\Om,
\exists \;  \text{a sequence} \; (\bsPhin)_n \subset C_0^1(\Om)\text{ with}\\ 
&\lim_{n\to\infty} \|(\bsPhin-\bsPsi)\chi_{\Om'}\|=0\text{ and } \lim_{n\to\infty} \|(\ID_0\bsPhin-\bsXi)\chi_{\Om'}\|=0\big\}.
\end{aligned}
\end{equation}
Here, the notation $\Om'\subset\!\subset\Om$ is used to denote the fact that $\Om'$ is a proper subset of $\Om$, 
i.e. that there exists a compact $K$ such that $\Om'\subset K\subset\Om$.

It is easy to see that $\mathcal F \subset\mathcal D(\ID_0^*)$.
Indeed, let $\bsPsi\in \mathcal F$. Fix $\Om'\subset\!\subset\Om$, take
$\bsPhi\in C_0^1(\Om)$ with $\text{supp}\,\bsPhi\subset\Om'$, and $(\bsPhin)_n,\bsXi$
as in\eqref{E:2.17}. Then
\begin{equation*}
\begin{aligned}
\dlab \ID_0\bsPhi,\bsPsi\drab
&=\dlab \ID_0\bsPhi,\bsPsi-\bsPhin\drab+\dlab \ID_0\bsPhi,\bsPhin\drab\\
&=\dlab \ID_0\bsPhi,\bsPsi-\bsPhin\drab+\dlab\bsPhi,\ID_0\bsPhin-\bsXi\drab+\dlab\bsPhi,\bsXi\drab\\
&=\dlab \chi_{\Om'}\ID_0\bsPhi,\bsPsi-\bsPhin\drab+\dlab\chi_{\Om'}\bsPhi,\ID_0\bsPhin-\bsXi\drab+\dlab\bsPhi,\bsXi\drab\\
&=\dlab \ID_0\bsPhi,(\bsPsi-\bsPhin)\chi_{\Om'}\drab+\dlab\bsPhi,(\ID_0\bsPhin-\bsXi)\chi_{\Om'}\drab+\dlab\bsPhi,\bsXi\drab\\
&\to \dlab\bsPhi,\bsXi\drab \quad\text{as }n\to\infty\,.
\end{aligned}
\end{equation*}
As $\Om'$ was arbitrary, this implies that
\begin{equation*}
\dlab \ID_0\bsPhi,\bsPsi\drab=\dlab\bsPhi,\bsXi\drab\quad\text{for all }\bsPhi\in C_0^1(\Om)\,,
\end{equation*}
and hence $\bsPsi\in\mathcal D(\ID_0^*)$, as claimed.

The remarkable fact proved by Friedrichs is that:
\begin{theorem}[\cite{Fr}]\label{T:M}
\begin{equation}\label{E:2.18}
\mathcal F=\mathcal D(\ID_0^*)\,.
\end{equation}
\end{theorem}

We now have the ingredients for the proof of Theorem~\ref{T:1}.
\begin{proof}[Proof of Theorem~\ref{T:1}] The proof proceeds in several steps.

\noindent
\textit{\underline{Step 1.} Part i. for $\bsE\equiv \bf 1$ and $\bsV\equiv 0$.}
Let $\bsPsi,\wti\bsPsi$ be weak solutions of $(\ID_0-i)\bsPsi=0$ and
$(\ID_0+i)\wti\bsPsi=0$, respectively i.e.
\begin{equation}\label{E:2.19}
\dlab(\ID_0+i)\bsPhi,\bsPsi\drab=\dlab(\ID_0-i)\bsPhi,\wti\bsPsi\drab=0\quad\text{for all }\bsPhi\in C_0^1(\Om)^k\,.
\end{equation}
From the basic criterion for (essential) self-adjointness \cite{AG,RS}, it is sufficient to prove that $\bsPsi=\wti\bsPsi=0$. 

Let $G_p$, $p=1,2,...$ as given by Lemma~\ref{L:4} . From Lemmas~\ref{L:3} and \ref{L:4} the operator of multiplication with 
$\bssig\big(\cdot,\nabla G_p\big)$ is bounded in $L^2(\Om)^k$ and
\begin{equation}\label{E:2.52}
\big\|\bssig(\cdot,\nabla G_p(\cdot))\big\|\leq \sup_{x\in\Om} \big|\nM G_p(x)\big|_{_{\mathcal M}}\,.
\end{equation}
From \eqref{E:2.13} :
\begin{equation}\label{E:2.20}
\big[G_p,\ID_0\big]=i\bssig\big(\cdot,\nabla G_p\big)\quad\text{on } C_0^1(\Om)^k\,.
\end{equation}

From \eqref{E:2.20},  \eqref{E:2.19} and the fact that $G_p\bsPhi \in C_0^1(\Om)^k$:
\begin{equation}\label{E:2.21}
\dlab(\ID_0+i)\bsPhi, G_p\bsPsi\drab=\dlab G_p(\ID_0-i)\bsPhi, \bsPsi\drab=
\dlab [G_p,\ID_0]\bsPhi, \bsPsi\drab=\dlab\bsPhi,{\bf F_p}\drab\ \,
\end{equation}
with
\begin{equation}\label{E:2.22}
{\bf F_p}=-i\bssig\big(\cdot,\nabla G_p(\cdot)\big)\bsPsi\in L^2(\Om)^k.
\end{equation}

The next step is to investigate the regularity properties of $\bsPsi$.
For all $q\geq 1$, we claim that 
\begin{equation}\label{E:2.23}
G_q\bsPsi\in\mathcal D(\bar \ID_0)\,,
\end{equation}
where $\mathcal D(\bar \ID_0)$ is the domain of the closure, $\bar \ID_0$, of $\ID_0$.
For, first note that the Euclidean topology on $\Om$ coincides with the topology 
induced by the distance function corresponding to \eqref{E:2.7} applied
to $\hatt{\bf M}$, see e.g. \cite[Corollary 1.4.1]{Jo}.
In particular, this means that saying that a set $K\subset\Om$ is compact means it
is compact in both topologies.

Let now $q\geq1$ be arbitrary, fixed. By Lemma~\ref{L:4}.i., $\text{supp}\,G_q$ is compact. By the Hopf-Rinow
theorem, there exists a compact $K$ in $\bsM$ such that $\text{supp}\,G_q\subset \text{Int}\, K$, and let $p_K$ be as in
Lemma~\ref{L:4}.ii..
By construction
\begin{equation}\label{SE:2.1}
0\leq G_q \leq \chi_{\text{Int}\,K} \leq G_{p_K}\leq 1,
\end{equation}
and
\begin{equation}\label{SE:2.2}
(1-\chi_{\text{Int}\,K})\bssig\big(\cdot,\nabla G_q(\cdot)\big)=0.
\end{equation}
In particular
\begin{equation}\label{SE:2.3}
G_{p_K}\big|_{_K}=1.
\end{equation}
By \eqref{E:2.22} and Lemma~\ref{L:4} we know that $iG_{p_K}\bsPsi+F_{p_K}\in L^2(\Om)^k$ .
Taking
$p=p_K$ in \eqref{E:2.21} yields
\begin{equation}\label{E:2.25}
\dlab\ID_0\bsPhi,G_{p_K}\bsPsi\drab=\dlab (\ID_0+i)\bsPhi,G_{p_K}\bsPsi\drab-i\dlab\bsPhi,G_{p_K}\bsPsi\drab
=\dlab\bsPhi,iG_{p_K}\bsPsi+{\bf F_{p_K}}\drab\,,
\end{equation}
which imply that $G_{p_K}\bsPsi \in \mathcal D(\ID_0^*)$ and then by Theorem \ref{T:M}
\begin{equation}\label{SE:2.4}
G_{p_K}\bsPsi \in \mathcal F.
\end{equation}

Taking $\Om'=\text{Int}\,K$ in \eqref{E:2.17} applied to $G_{p_K}\bsPsi$, it follows that there exists $\bsXi\in L^2(\Om)^k$ and a sequence 
$(\bsPhin)_n\subset C_0^1(\Om)^k$ such that
\begin{equation}\label{E:2.27}
\lim_{n\to\infty}\| \chi_{\text{Int}\,K}(\bsPhin-G_{p_K}\bsPsi)\|= 
\lim_{n\to\infty} \|\chi_{\text{Int}\,K}(\ID_0\bsPhin-\bsXi)\|=0\,.
\end{equation}

Denote $\wti\bsPhin=G_q\bsPhin$. We claim that as $n \rightarrow \infty$, $\wti\bsPhin\to G_q\bsPsi$ and the sequence $(\ID_0\wti\bsPhin)_n$ is convergent, that is \eqref{E:2.23} holds true.
 Indeed, from the definition of  $\wti\bsPhin$,  \eqref{SE:2.1},  \eqref{SE:2.3} and \eqref{E:2.27} we obtain that
 
\begin{equation}\label{E:2.29/30}
\begin{aligned}
\lim_{n\to\infty} \|\wti\bsPhin-G_q\bsPsi\| 
&=\lim_{n\to\infty} \|G_q(\bsPhin-\bsPsi)\| \leq\lim_{n\to\infty} \|\chi_{\text{Int}\,K}(\bsPhin-\bsPsi)\|\\
&=\lim_{n\to\infty} \|\chi_{\text{Int}\,K}(\bsPhin-G_{p_K}\bsPsi)\|=0
\end{aligned}
\end{equation}
which proves the first part of the claim above. Turning to $\ID_0\wti\bsPhi$, 
a direct computation using \eqref{E:2.9}  yields
\begin{equation}\label{E:2.28}
\begin{aligned}
\ID_0 G_q\bsPhin&=G_q\ID_0\bsPhin-i\bssig(\cdot,\nabla G_q(\cdot))\bsPhin\\
&=G_q(\ID_0\bsPhin-\bsXi)+G_q\bsXi-i\bssig(\cdot,\nabla G_q(\cdot))(\bsPhin-\bsPsi)+i\bssig(\cdot,\nabla G_q(\cdot))\bsPsi\,.
\end{aligned}
\end{equation}
Now from \eqref{SE:2.1} and \eqref{E:2.27}
\begin{equation}\label{SE:2.5}
\lim_{n\to\infty} \|G_q(\ID_0\bsPhin-\bsXi)\| \leq \lim_{n\to\infty} \|\chi_{\text{Int}\,K}(\ID_0\bsPhin-\bsXi)\|=0.
\end{equation}
Further, from \eqref{SE:2.1}, \eqref{SE:2.2}, \eqref{SE:2.3} 
\begin{equation}
\begin{aligned}
 \|\bssig(\cdot,\nabla G_q(\cdot))\,(\bsPhin-\bsPsi)\|&
 = \|((1-\chi_{\text{Int}\,K}) +\chi_{\text{Int}\,K})\bssig(\cdot,\nabla G_q(\cdot))\,(\bsPhin-\bsPsi)\| \\
&= \|\chi_{\text{Int}\,K}\bssig(\cdot,\nabla G_q(\cdot))\,(\bsPhin-\bsPsi)\| \\ 
& = \|\bssig(\cdot,\nabla G_q(\cdot))\,\chi_{\text{Int}\,K}(\bsPhin-G_{p_K}\bsPsi)\| \\
&\leq \|\bssig(\cdot,\nabla G_q(\cdot))\|\;\|\chi_{\text{Int}\,K}(\bsPhin-G_{p_K}\bsPsi)\|
\end{aligned}
 \end{equation}
which together with \eqref{E:2.52} and \eqref{E:2.27} leads to 
\begin{equation}\label{SE:2.6}
\lim_{n\to\infty} \|\bssig(\cdot,\nabla G_q(\cdot))\,(\bsPhin-\bsPsi)\|=0.
\end{equation}
Putting together  \eqref{E:2.28}, \eqref{E:2.5} and \eqref{E:2.6}  one obtains that
\begin{equation}\label{E:2.35}
\lim_{n\to\infty} \|\ID_0 G_q\bsPhin-G_q\bsXi-i\bssig(\cdot,\nabla G_q(\cdot))\bsPsi\|=0\,,
\end{equation}
which together with \eqref{E:2.29/30} shows that \eqref{E:2.23} holds true, thus proving our claim.

Given that $\bar\ID_0$ is symmetric, the claim \eqref{E:2.23} implies that for any $p\geq 1$ and $\bsPhi\in C_0^1(\Om)^k$
\begin{equation}\label{E:2.36}
\dlab (\bar\ID_0+i)\bsPhi, G_p\bsPsi\drab=\dlab \bsPhi,(\bar\ID_0-i) G_p\bsPsi\drab\,,
\end{equation}
while from \eqref{E:2.21} and \eqref{E:2.22} we find that
\begin{equation}\label{E:2.37}
\dlab (\bar\ID_0+i)\bsPhi, G_p\bsPsi\drab=\dlab (\ID_0+i)\bsPhi, G_p\bsPsi\drab
=-i\dlab\bsPhi,\bssig(\cdot,\nabla G_p(\cdot))\bsPsi\drab\,.
\end{equation}
The last two equations together give that
\begin{equation}\label{E:2.38}
(\bar\ID_0-i) G_p\bsPsi=-i\bssig(\cdot,\nabla G_p(\cdot))\bsPsi\,.
\end{equation}

Now let $\wti K$ be an arbitrary compact in $\Om$. For $p\geq p_{\wti K}$, 
the symmetry of $\bar\ID_0$ together with Lemma~\ref{L:4}ii. and \eqref{E:2.38} imply that
\begin{equation}\label{E:2.39}
\big\|\bssig(\cdot,\nabla G_p(\cdot))\bsPsi\big\|=\big\|(\bar\ID_0-i) G_p\bsPsi\big\|
\geq \| G_p\bsPsi\|\geq\|\chi_{\wti K}\bsPsi\|\,.
\end{equation}
From \eqref{E:2.52}, it easily follows that
\begin{equation}
\big\|\bssig(\cdot,\nabla G_p(\cdot))\bsPsi\big\|\leq \sup_{x\in\Om} \big|\nM G_p(x)\big|_{_{\mathcal M}}\,\|\bsPsi\|\,,
\end{equation}
which together with \eqref{E:2.39} gives
\begin{equation}\label{E:2.40}
\|\chi_{\wti K}\bsPsi\|\leq \sup_{x\in\Om} \big|\nM G_p(x)\big|_{_{\mathcal M}}\,\|\bsPsi\|\,.
\end{equation}
Taking $p\to\infty$ and using Lemma~\ref{L:4}iii., we obtain that $\|\chi_{\wti K}\bsPsi\|=0$,
which, given that $\wti K$ is arbitrary, implies $\bsPsi=0$, as required.
Repeating this argument for $\wti\bsPsi$, we obtain that $\wti\bsPsi=0$, thus completing the proof.

\medskip

\noindent
\textit{\underline{Step 2.} Part ii. for $\bsE\equiv \bf1$.} The proof is essentially the same as 
in Step 1, but instead of \eqref{E:2.23} we use the assumption that $\ID_0+\boldsymbol{V_1}+\chi_{\supp G_p} \boldsymbol{V_2}$ 
is essentially self-adjoint. Indeed, using the fact that $[G_p,\boldsymbol{V_1}+\boldsymbol{V_2}]=0$ for all $p\geq 1$ and
that $\chi_{\supp G_p} \boldsymbol{V_2}=\boldsymbol{V_2}$ on $\supp G_p$, we obtain that
\begin{equation}\label{E:2.42/43}
\dlab (\ID_0+\boldsymbol{V_1}+\chi_{\supp G_p} \boldsymbol{V_2}+i)\bsPhi,G_p\bsPsi\drab
=\dlab (\ID_0+\boldsymbol{V_1}+ \boldsymbol{V_2}+i)\bsPhi,G_p\bsPsi\drab=\dlab \bsPhi, F_p\drab\,,
\end{equation}
with $F_p$ given by \eqref{E:2.22}. From this and the essential self-adjointness of
$\ID_0+\boldsymbol{V_1}+\chi_{\supp G_p} \boldsymbol{V_2}$ it follows that
\begin{equation}\label{E:2.44}
G_p\bsPsi\in \mathcal D\big((\ID_0+\boldsymbol{V_1}+\chi_{\supp G_p} \boldsymbol{V_2})^*\big)
=\mathcal D\big(\overline{\ID_0+\boldsymbol{V_1}+\chi_{\supp G_p} \boldsymbol{V_2}}\big)\,,
\end{equation}
and hence
\begin{equation}
\big(\overline{\ID_0+\boldsymbol{V_1}+\chi_{\supp G_p} \boldsymbol{V_2}}-i\big)G_p\bsPsi=-i\bssig(\cdot,\nabla G_p(\cdot))\bsPsi\,,
\end{equation}
so that 
\begin{equation}
\|G_p\bsPsi\|\leq\|\bssig(\cdot,\nabla G_p(\cdot))\bsPsi\|\,.
\end{equation}
The rest of the argument remains unchanged.

\medskip
\noindent
\textit{\underline{Step 3.} Part i. for $\bsE\equiv \bf 1$ and general $\bsV$.} 
Since $\bsV$ is assumed to be in $L^\infty_\text{loc}(\Om)$, it follows that for every compact
set $K\subset\Om$ we have $\chi_K \bsV\in L^\infty(\Om)$, so the essential self-adjointness of 
$\ID_0$ (proven in Step 1) together with the Kato-Rellich theorem (see, e.g., \cite{RS})
imply the essential self-adjointness of $\ID_0+\chi_K \bsV$. An application of Theorem~\ref{T:1}ii. with $\bsE\equiv\bf 1$ (i.e. as proven in Step 2) and $\boldsymbol{V_1}=0$, $\boldsymbol{V_2}=\bsV$ yields the desired result.

\medskip
\noindent
\textit{\underline{Step 4.} Parts i. and ii. for general $\bsE$.} This case reduces to the 
"canonical form" corresponding to $\bsE\equiv \bf 1$ by a well-known transformation
(see, e.g., \cite[Chap. 3, \S 5]{J}).

Let $\mathbb{S}\,:\, \mathcal H_\bsE\,\to\, L^2(\Om)^k$ be given by
\begin{equation}\label{E:2.45}
\big(\mathbb{S}\bsPsi\big)(x)=\bsE(x)^{1/2}\bsPsi(x)\,.
\end{equation}
Then $\mathbb{S}$ is unitary and
\begin{equation}\label{E:2.46}
\mathbb{S }C_0^1(\Om)^k=\mathbb{S}^{-1} C_0^1(\Om)^k=C_0^1(\Om)^k\,.
\end{equation}
Now consider the symmetric operator $\wti{\ID}$ in $L^2(\Om)^k$ defined as
\begin{equation}\label{E:2.47}
\wti{\ID}=\mathbb{S}\ID \mathbb{S}^{-1},\qquad \mathcal D(\wti{\ID})=C_0^1(\Om)^k\,.
\end{equation}
By a direct computation, we find that
\begin{equation}\label{E:2.48}
\wti{\ID}=\wti{\ID}_0+\wti{\IV}_0+\wti{\IV}\,,
\end{equation}
where
\begin{equation}\label{E:2.49}
\begin{aligned}
\wti{\ID}_0&=\tfrac12\big(\vec{\wti A} \cdot\bsD+\bsD\cdot\vec{\wti A}\big)\\
\wti{\IV}_0&=\tfrac12\big(\bsE^{-1/2}\bsA\cdot(\bsD\bsE^{-1/2})
-(\bsD\bsE^{-1/2})\cdot \bsA \bsE^{-1/2}\big)\\
\wti{\IV}&=\bsE^{-1/2}\IV \bsE^{-1/2}\\
\wti{A}^j&=\bsE^{-1/2}\boldsymbol{A^j}\bsE^{-1/2}\,,\quad (\bsD\bsE^{-1/2})^j=-i\frac{\partial (\bsE^{-1/2})}{\partial x_j}\text{ for all }1\leq j\leq d\,.
\end{aligned}
\end{equation}

Assume that $\bsV\in L^\infty_\text{loc}(\Om)$. Then by \eqref{E:2.49}
we find that $\wti{\IV}_0,\wti{\IV}\in L^\infty_\text{loc}(\Om)$, and hence from
Theorem~\ref{T:1}i. for $\bsE\equiv \mathds{1}$ (Step 3), we obtain that
$\wti{\ID}$ is essentially self-adjoint in $L^2(\Om)^k$. 
By invariance of essential self-adjointness under unitary transformations, this
implies that $\ID$ is essentially self-adjoint in $\mathcal H_\bsE$, 
completing the proof of part i. of the Theorem.

Now suppose that the hypotheses of part ii. hold. Then for any compact set $K\subset\Om$
we know that $\wti{\ID}_0+\wti{\IV}_0+\wti{\IV}_1+\chi_K\wti{\IV}_2$
is essentially self-adjoint in $L^2(\Om)^k$. From Theorem \ref{T:1}ii. in the case $\bsE ={\bf 1}$
(Step 3) it follows that $\wti{\ID}$ is also essentially self-adjoint, 
and unitary invariance again implies that $\ID$ is essentially self-adjoint in $\mathcal H_\bsE$,
thus completing the proof of this step and the theorem.
\end{proof}

\section{Consequences}\label{S:3}

In this section we give some corollaries of Theorem~\ref{T:1}.
We begin with the fact that in the case $\bsM>0$, Theorem \ref{T:1} simplifies to:
\begin{corollary}\label{T:1e}
In the setting of Theorem \ref{T:1}, assume that $\bsM>0$ on $\Om$ and that 
$\mathcal M$ is a complete metric space. 
Then the conclusions of Theorem \ref{T:1} hold true.
\end{corollary}

Corollary~\ref{T:1e} follows from Theorem~\ref{T:1} via the following density result
due to Agmon \cite[Lemma~A1]{Ag}:
 
 \begin{lemma}\label{L:A1}
Let $\bsG(x)$ be a continuous  $d\times d$ matrix-valued function on $\Om$, $\bsG(x)=\overline{\bsG(x)}>0$
for all $x\in\Om$. Then for every $\delta>0$ there exists a $d\times d$ matrix-valued function 
$\bsH\in C^\infty(\Om)$, $\bsH(x)=\overline{\bsH(x)}$ such that
\begin{equation}\label{E:2.50}
(1-\delta)\,\bsG(x)\leq \bsH(x)\leq (1+\delta)\, \bsG(x)\qquad\text{for all } x\in\Om\,.
\end{equation}
\end{lemma}

\begin{proof}[Proof of Corollary~\ref{T:1e}]
Since in this case $2\bsM$ satisfies the hypotheses of Lemma~\ref{L:A1}, it follows that
for $\delta=\frac12$ there exists $\hatt \bsM\in C^\infty(\Om)$ such that
\begin{equation}\label{E:2.51}
\bsM(x)\leq \hatt \bsM(x)\leq 3\bsM(x)\qquad\text{for all } x\in\Om\,.
\end{equation}
From \eqref{E:2.51} we see that $\hatt \bsM^{-1}\geq \frac13 \bsM^{-1}$, and hence since $\mathcal M$
is metrically complete, $\hatt{\mathcal M}$ is metrically complete and then the Hopf-Rinow theorem implies that $\cMh$
is a complete Riemannian manifold. This together with the first inequality in \eqref{E:2.51}
allow us to apply Theorem~\ref{T:1}, which concludes the proof.
\end{proof}

The next straightforward statement shows how we often apply
the result of Theorem~\ref{T:1}i.

\begin{corollary}\label{C:0}
In the hypotheses of Theorem ~\ref{T:1} replace \eqref{E:2.7.1} by:
\begin{equation}\label{E:3.0}
\bsM(x)\leq a\,\Mh(x)\qquad\text{ for some} \quad a>0  \quad \text {and  for all }x\in\Om\,.
\end{equation}
If $\bsV\in L^\infty_{loc}(\Om)$, then $\ID$ is essentially self-adjoint.
\end{corollary}

\begin{proof}
The proof of this statement consists of merely observing that, given that $a>0$, 
the completeness of $\cMh$ is equivalent to the completeness
of $\cMh_a$, the Riemannian manifold associated to $\Mh_a(x)=a\Mh(x)$, $x\in\Om$.
\end{proof}

While this corollary is clearly equivalent to the statement of Theorem~\ref{T:1}i., 
we give it here since, when searching for $\Mh$ in applications, we often
come to inequalities of the type of \eqref{E:3.0} rather then the "pure" \eqref{E:2.7.1}.

\begin{corollary}\label{C:1}
Let $\Om=\IR^d$. In the setting of Corollary~\ref{C:0}, assume that there exist $R\geq 1$, 
$\kappa<\infty$, and $d\,:\,[R,\infty)\to(0,\infty)$ continuous such that
\begin{equation}\label{E:3.1}
\bsM(x)\leq \kappa d(|x|)^2\,{\bf 1}\qquad\text{for all } |x|\geq R
\end{equation}
and
\begin{equation}\label{E:3.2}
\int_R^\infty \frac{1}{d(t)}\,dt=\infty\,.
\end{equation}
Then $\ID$ is essentially self-adjoint.
\end{corollary}

\begin{proof}
For $x\in\IR^d$, define
\begin{equation}\label{E:3.3}
\wti \bsM(x)=
\begin{cases}
d(R)^2 {\bf 1} &\quad\text{for }|x|\leq R\\
d(|x|)^2{\bf 1} &\quad\text{for }|x|>R\,.
\end{cases}
\end{equation}
From \eqref{E:3.1}, it follows that for all $x\in\IR^d$
\begin{equation}\label{E:3.4}
\bsM(x)\leq \max \big\{\sup_{|x|\leq R} \|\bsM(x)\|_{\mathcal B(\IC^d)},\, d(R)^{-2},\,\kappa\big\}\,\wti \bsM(x)\,,
\end{equation}
while from Lemma~\ref{L:A1} we conclude that there exists $\hatt \bsM\in C^\infty(\IR^d)$ such that
\begin{equation}\label{E:3.5}
\frac12\,\wti \bsM(x)\leq\hatt \bsM(x)\leq\frac32\,\wti \bsM(x)\,.
\end{equation}

From \eqref{E:3.4} and the first inequality in \eqref{E:3.5} we find that
\begin{equation*}
\bsM(x)\leq 2  \max \big\{\sup_{|x|\leq R} \|\bsM(x)\|_{\mathcal B(\IC^d)},\, d(R)^{-2},\,\kappa\big\}\,\hatt \bsM(x)
\quad\text{for all } x\in\IR^d\,,
\end{equation*}
thus satisfying \eqref{E:3.0} with 
\begin{equation*}
a=2  \max \big\{\sup_{|x|\leq R} \|\bsM(x)\|_{\mathcal B(\IC^d)},\, d(R)^{-2},\,\kappa\big\}>0\,.
\end{equation*}
Furthermore, the second inequality in \eqref{E:3.5} and \eqref{E:3.3} imply that
\begin{equation*}
\Mh(x)\leq\frac32\, d(|x|)^2\qquad\text{for all } |x|\geq R\,,
\end{equation*}
which combined with hypothesis \eqref{E:3.2} shows that $\cMh$ is complete. 
The conclusion then follows by Corollary~\ref{C:0}.
\end{proof}

\begin{corollary}\label{C:2}
Let $d=1$ and $\Om=\IR$. In the setting of Corollary~\ref{C:0} assume that
there exists $R<\infty$ such that $\text{det}\,{\bf A^1}(x)\neq 0$ for all $|x|\geq R$. 
If
\begin{equation}\label{E:3.6}
\int_{\pm(R+1)}^{\pm\infty} \frac{dx}{\big\|\bsE(x)^{-1/2}{\bf A^1}(x)\bsE(x)^{-1/2}\big\|_{\mathcal B(\IC^k)}}=\infty
\end{equation}
then $\ID$ is essentially self-adjoint.
\end{corollary}

\begin{proof}
Since $d=1$, $\bsM=M$ is scalar, non-negative valued. Note that 
since $\text{det}\,{\bf A^1}(x)\neq 0$ for all $|x|\geq R$, both $M(R+1)$ and
$M(-R-1)$ are positive.
Let then
\begin{equation*}
\eta=\min\big\{M(-R-1), M(R+1)\big\}>0\,,
\end{equation*}
\begin{equation*}
A_{\eta}=\bigg\{x\in[-R-1,R+1]\,\big|\, M(x)\leq\frac{\eta}{2}\bigg\}\,,
\end{equation*}
and
\begin{equation}\label{E:3.7}
\wti M(x)=
\begin{cases}
\frac{\eta}{2} &\quad\text{for }x\in A_{\eta},\\
  & \\
M(x) &\quad\text{otherwise}\,.  
\end{cases}
\end{equation}

By construction, $\wti M$ is positive on $\IR$ and is continuous, 
since $A_\eta$ is a closed set inside $(-R-1,R+1)$ and $M$ itself
is continuous. Furthermore,
\begin{equation}\label{E:3.8}
M(x)\leq \wti M(x)\quad\text{for all }x\in\IR\,.
\end{equation} 
From \eqref{E:3.6}, \eqref{E:3.7} and the fact (see \eqref{E:2.6}) that on $\IR$
\begin{equation*}
M
=\text{Tr}\big(\bsE^{-1/2}{\bf A^1} \bsE^{-1/2}\big)^2\leq k
\big\|\bsE^{-1/2}{\bf A^1}\bsE^{-1/2}\big\|^2_{\mathcal B(\IC^k)}\,,
\end{equation*}
we conclude that $\wti{\mathcal M}$ is metrically complete. 
As in the previous proofs, from Lemma~\ref{L:A1} used for $\wti M$ 
and $\delta=\frac12$, we know that there exists $\Mh\in C^\infty(\IR)$
such that
\begin{equation*}
\frac12 \wti M\leq \Mh\leq\frac32\wti M\quad\text{on }\IR\,.
\end{equation*}
We then find that $M\leq 2\Mh$ on $\IR$, and that $\cMh$
is metrically complete, so the conclusion follows by Corollary~\ref{C:0}.
\end{proof}

For the particular case when ${\bf A^1}(x)$ is invertible for every $x$ (but with weaker
smoothness conditions), Corollary~\ref{C:2} was proven by Lesch and Malamud 
\cite[Theorem 3.8]{LM}. 

We end this section with a discussion of relation between Chernoff, Fattorini conditions (see \eqref{Ch1}- \eqref{AF2}) and the completeness condition in Theorem \ref{T:1}. For, let:
\begin{equation}\label{CFN}
c(x)=  \sup_{\xi\in\IR^d,|\xi|=1} 
\bigg\|\sum_{j=1}^d \xi_j\bsAj(x)\bigg\|_{\mathcal B(\IC^k)}\,, \quad r(x)= \text{max}_j\{\|\bsA^j(x)\|\}.
\end{equation}
Then one can easily see that
\begin{equation}\label{CF}
r(x)\leq c(x) \leq d^{1/2}r(x)
\end{equation}
which implies the equivalence between Chernoff and Fatorini conditions. Further
recall that for $\xi\in\IR^d$, $\bssig(x,\xi)$ is Hermitian, and so \eqref{E:2.12} implies that
\begin{equation*}
\langle \xi, \bsM(x)\xi\rangle=\text{Tr}\,\bssig(x,\xi)^2\,.
\end{equation*}
Hence
\begin{equation*}
\big\|\bssig(x,\xi)\big\|_{\mathcal B(\IC^k)}^2\geq\frac{1}{k} \text{Tr}\,\bssig(x,\xi)^2
=\frac{1}{k}\langle \xi, \bsM(x)\xi\rangle\,,
\end{equation*}
and then
\begin{equation}\label{E:3.11}
 k c(x)^2{\bf 1}\geq \bsM(x)\quad\text{for all } x\in\IR^d\,
\end{equation}
which together with Corollary \ref{C:1} shows that the completeness condition in Theorem \ref{T:1}
is a weaker, ``anisotropic" form of Chernoff and Fattorini conditions.  
One can easily construct examples where both 
Chernoff and Fattorini conditions are inconclusive and still Theorem \ref{T:1} ensures essential self-adjointness.

\section{Applications to wave propagation problems of classical physics}\label{S:4}

We now focus on the application of the main result to the problem of energy confinement for some wave 
propagation phenomena in classical physics, namely in electromagnetism and continuum mechanics.
In the linear approximation, these phenomena are described by 
hyperbolic equations which can be cast in the form \cite{Wi1}:
\begin{equation}\label{E:3.12}
i\frac{\partial}{\partial t}\bsPsi=\ID_0\bsPsi\,.
\end{equation}
The important fact for the physical interpretation is that $\|\bsPsi\|_\bsE^2$ is (up to a multiplicative constant) nothing but
the total energy of the system in the state $\bsPsi$. 

\begin{example}[Telegraph equation]
We begin with the most elementary example: the propagation of electric
signals along a loosless line (telegraph equation) \cite{Wi1}. The
equations are
\begin{equation}\label{E:star1}
\begin{aligned}
L(x)\,\frac{\partial i(x,t)}{\partial t}+\frac{\partial v(x,t)}{\partial x}&=0\\
C(x)\,\frac{\partial v(x,t)}{\partial t}+\frac{\partial i(x,t)}{\partial x}&=0
\end{aligned}
\end{equation} 
for all $x\in(a,b)\subset\IR$ and $t\in\IR$. Here $i$ and $v$ denote
current and voltage in the line.  $L(x)>0$, $C(x)>0$ are
densities of inductance and capacitance, respectively and are supposed to to be in $C^1(\Om)$. The energy density is
\begin{equation}\label{E:star2}
e(x,t)=\frac12\big(L(x)i(x,t)^2+C(x)v(x,t)^2\big)\,,
\end{equation}
and the local velocity of the signal is (in appropriate units)
\begin{equation}\label{E:star3}
c(x)=\big(L(x)C(x)\big)^{-1/2}\,.
\end{equation}

The equations \eqref{E:star1} can be written in the form \eqref{E:3.12} for
\begin{equation}\label{E:star4}
\bsPsi(x,t)=\big[i(x,t)\,,\,v(x,t)\big]^T
\end{equation}
with
\begin{equation}\label{E:star5}
\bsE(x)=\left[\begin{matrix}
L(x) & 0\\
0 & C(x)
\end{matrix}\right]\,,\quad
\boldsymbol{A^1}=\left[\begin{matrix}
0 & 1\\
1 & 0
\end{matrix}\right]\,.
\end{equation}
Notice that
\begin{equation}\label{E:star6}
\dlab\bsPsi(t),\bsPsi(t)\drab_{\bsE}=2\int_a^b e(y,t)\,dy
\end{equation}
i.e. (up to a multiplicative constant) $\|\bsPsi(t)\|^2_\bsE$ is the total energy.
In this case, the velocity matrix is just a positive number, and plugging
\eqref{E:star5} into \eqref{E:2.6} one obtains (see \eqref{E:star3})
\begin{equation}\label{E:star7}
M(x)=\frac{2}{L(x)C(x)}=2c(x)^2\,.
\end{equation}

The metric completeness of $\big((a,b),M^{-1}\big)$, which via Corollary~\ref{T:1e}
ensures the self-adjointness of $\ID_0$ and hence the conservation of total energy,
is equivalent with
\begin{equation}\label{E:star8}
\int_a\frac{1}{c(y)}\,dy=\infty\quad\text{and}\quad\int^b\frac{1}{c(y)}\,dy=\infty\,,
\end{equation}
saying that if $\Psi(\cdot,0)$ is compactly supported, then for all $t\in\IR$ 
$\Psi(\cdot,t)$ is also compactly supported.
\end{example}

\begin{example}[Maxwell's equations \cite{Wi1}] \label{E:1}
The next example we consider here are the Maxwell equations without sources in a linear,
nondispersive, inhomogeneous anisotropic medium which fills a domain $\Om\subset\IR^2$.
The properties of the medium are described by its dielectric permitivity and magnetic permeability tensors
\begin{equation*}
\bsvareps(x)=\big(\varepsilon_{jk}(x)\big)_{1\leq j,k\leq 3}\text{ and }
{\bsmu}(x)=\big(\mu_{jk}(x)\big)_{1\leq j,k\leq 3}\,,\quad x\in\Om
\end{equation*}
respectively, and the state of the system is given by electric and magnetic fields 
$\bscE(x,t)=\big(\mathcal E_j(x,t)\big)_{1\leq j\leq3}$ and 
$\bscH(x,t)=\big(\mathcal H_j(x,t)\big)_{1\leq j\leq3}$, respectively.
We assume that the (real) matrices $\bsvareps,\bsmu\in C^1(\Om)$ are strictly positive definite
at every $x\in\Om$.

The equations of motion for $\bscE$ and $\bscH$ are
\begin{equation}\label{E:3.13}
\begin{aligned}
\big(\nabla\times\bscH\big)(x,t) - \bsvareps(x)\,
\frac{\partial\bscE(x,t)}{\partial t}&=0\\
\big(\nabla\times\bscE\big)(x,t)+ \bsmu(x)\,
\frac{\partial\bscH(x,t)}{\partial t}&=0
\end{aligned}
\end{equation}
for all $x\in\Om$ and $t>0$. 
These can be rewritten in the form \eqref{E:3.12} for the vector
\begin{equation}
\label{E:3.13*}
\bsPsi(x,t)=\,\big(\cE_1(x,t),\cE_2(x,t),\cE_3(x,t),\cH_1(x,t),\cH_2(x,t),\cH_3(x,t)\big)^T\,.
\end{equation}
Here $\bsE(x)$ is the $6\times6$ strictly positive matrix
\begin{equation}\label{E:3.14}
\bsE(x)=
\left[\begin{matrix}
\bsvareps & 0\\
0 & \bsmu
\end{matrix}\right]\,,
\end{equation}
and $\bsAj$ are the real, constant $6\times6$  matrices
\begin{equation}\label{E:3.15.1}
\bsAj=-\left[\begin{matrix}
0 & \bsaj\\
(\bsaj)^T & 0
\end{matrix}\right]\,,
\end{equation}
with
\begin{equation}\label{E:3.15.2}
\boldsymbol{a^1}=
\left[\begin{matrix}
0&0&0\\
0&0&-1\\
0&1&0
\end{matrix}\right]\,,\quad
\boldsymbol{a^2}=
\left[\begin{matrix}
0&0&1\\
0&0&0\\
-1&0&0
\end{matrix}\right]\,,\quad
\boldsymbol{a^3}=
\left[\begin{matrix}
0&-1&0\\
1&0&0\\
0&0&0
\end{matrix}\right]\,.
\end{equation}

The energy density is given
(up to a multiplicative constant) by
\begin{equation}\label{E:3.16}
e(x,t)=\bscE(x,t)^T\bsvareps(x)\bscE(x,t)+\bscH(x,t)^T\bsmu(x)\bscH(x,t)\,,
\end{equation}
and so from \eqref{E:3.13}--\eqref{E:3.16} we obtain that
\begin{equation}\label{E:3.17}
\dlab\bsPsi(t),\bsPsi(t)\drab_{_\bsE}=\int_\Om e(x,t)\,dx\,.
\end{equation}
A straightforward  computation which uses  the invariance of the trace under cyclic permutations and transposition shows
 that the (velocity) matrix
$\bsM(x)$ is given by
\begin{equation}\label{E:3.18}
M_{jl}=
2\,\text{Tr}\big(\bsvareps^{-1/2}\bsaj\bsmu^{-1/2}(\bsvareps^{-1/2}\bsal\bsmu^{-1/2})^T\big)
\end{equation}
The physical meaning of $\bsM(x)$ becomes clear in the isotropic case, i.e. the case where there
exist scalar-valued functions $\varepsilon$ and $\mu$ such that:
\begin{equation}\label{E:3.19}
\bsvareps(x)=\varepsilon(x){\bf{1}}\quad\text{and}\quad \bsmu(x)=\mu(x){\bf{1}}
\end{equation}
for scalar-valued functions $\varepsilon$ and $\mu$; when \eqref{E:3.18} reduces to
\begin{equation}\label{E:3.20}
\bsM(x)=4c(x)^2{\bf{1}}
\end{equation}
where
\begin{equation}\label{E:3.21}
c(x)=\big(\varepsilon(x)\mu(x)\big)^{-1/2}
\end{equation}
is the local velocity of light, hence the condition for essential self-adjointness
of $\ID_0$ translates to a condition on $c(x)$ only. An explicit form of such
a condition is given by Corollary~\ref{C:1}. Mimicking that proof, one can also
show that if $\Om$ is bounded, a sufficient condition for the essential self-adjointness
of $\ID_0$ is to have $c(x)\to0$ sufficiently fast as $x\to\partial\Om$:

\begin{corollary}\label{C:4}
In the setting of Theorem~\ref{T:1}i. assume that $\Om$ is bounded, 
and for $x\in\Om$ let $\delta(x)=\text{dist}(x,\partial\Om)$. Further assume that
there exist $\kappa<\infty$ and $b\,:\,(0,1)\to(0,\infty)$ continuous such that
\begin{equation}\label{E:3.22}
c(x)\leq \kappa b\big(\delta(x)\big)\quad\text{with }\lim_{\eta\to0+} \int_\eta \frac{dy}{b(y)}=\infty\,.
\end{equation}
Then $\ID_0$ is essentially self-adjoint.
\end{corollary} 
\end{example}

\begin{example}[Acoustic waves \cite{P,Wi1}] \label{E:2}
The next example is the propagation of acoustic waves in an anisotropic
inhomogeneous medium at rest. Acoustics is ordinarily concerned with
small disturbances, hence a linear (elastic) approximation is applicable.
The equation of motion for the displacement $\bsxi(x,t)$ of the ``particle"
(i.e., infinitesimal volume of medium) nominally at $x$ reads (Newton's law)
\begin{equation}\label{E:3.23}
\rho(x)\,\frac{\partial^2\xi_j(x,t)}{\partial t^2}=\sum_{l=1}^3 \frac{\partial\sigma_{jl}(x,t)}{\partial x_l}
\quad 1\leq j\leq3\,,
\end{equation}
where $\rho(x)>0$ is the density of the medium at rest and $\sigma_{jl}(x)=\sigma_{lj}(x)$ is the stress tensor.
In the linear approximation, the stress tensor is linearly dependent on the strain tensor
\begin{equation}\label{E:3.24}
\epsilon_{jl}(x,t)=\frac12\left(\frac{\partial\xi_j(x,t)}{\partial x_l}+\frac{\partial \xi_l(x,t)}{\partial x_j}\right)\,.
\end{equation}
The linear dependence (generalized Hooke's law) is
\begin{equation}\label{E:3.25}
\sigma_{jl}(x,t)=\sum_{m,n=1}^3 c_{jl,mn}(x)\epsilon_{mn}(x,t)\,,
\end{equation}
where the stress-strain tensor $c_{jl,mn}$ encodes the elastic properties of the medium. From
general physical requirements, $c_{jl,mn}$ has the symmetries
\begin{equation}\label{E:3.26}
c_{jl,mn}=c_{jl,nm}=c_{lj,mn}=c_{mn,jl}\,,
\end{equation}
hence out of its 81 components, only 21 are independent. The total mechanical density energy 
(taken to be zero at rest) as given by the sum of kinetic and potential energy is
\begin{equation}\label{E:3.27}
e(x,t)=\frac12\rho(x)\sum_{j=1}^3 \left(\frac{\partial\xi_j(x,t)}{\partial t}\right)^2
+\frac12\sum_{j,l=1}^3 \sigma_{jl}(x,t) \epsilon_{jl}(x,t)\,.
\end{equation}

In order to rewrite the equation of motion in the form \eqref{E:3.12}, it is convenient
to introduce the stiffness $6\times6$ (symmetric) matrix
\begin{equation}\label{E:3.28}
\bsC=
\left[\begin{matrix}
c_{11,11} & c_{11,22} & c_{11,33} & c_{11,12} & c_{11,23} & c_{11,31} \\
c_{22,11} & c_{22,22} & c_{22,33} & c_{22,12} & c_{22,23} & c_{22,31} \\
\vdots &\vdots & \vdots & \vdots & \vdots & \vdots \\
c_{31,11} & c_{31,22} & c_{31,33} & c_{31,12} & c_{31,23} & c_{31,31}
\end{matrix}\right]
\end{equation}
and the vectors
\begin{equation}\label{E:3.29}
\boldsymbol\Sigma=\big[\sigma_{11},\,\sigma_{22},\,\sigma_{33},\,\sigma_{12},\,\sigma_{23},\,\sigma_{31}\big]^T
\quad\text{and}\quad
\bsP=\left[\rho\frac{\partial\xi_1}{\partial t},\,\rho\frac{\partial\xi_2}{\partial t},
\,\rho\frac{\partial\xi_3}{\partial t}\right]^T\,.
\end{equation}
Now note that, by using \eqref{E:3.25}, the potential energy density takes the form
\begin{equation}\label{E:3.30}
\frac12\sum_{j,l=1}^3 \sigma_{jl}\epsilon_{jl}=\big(\bscE,\bsC\bscE\big)\,,
\end{equation}
where 
\begin{equation}\label{E:3.31}
\bscE=\big[\epsilon_{11},\,\epsilon_{22},\,\epsilon_{33},\,\epsilon_{12},\,\epsilon_{23},\,\epsilon_{31}\big]^T\,.
\end{equation}
The physical requirement that the potential energy density is positive and
has a unique global minimum at $\bscE=0$ implies that
\begin{equation}\label{E:3.32}
\bsC(x)>0\quad\text{for all }x\in\Om\,.
\end{equation}
In terms of $P_j=\rho\frac{\partial\xi_j}{\partial t}$ and $\sigma_{jl}$,
\eqref{E:3.23} takes the form
\begin{equation}\label{E:3.33}
\frac{\partial P_j}{\partial t}=\sum_{l=1}^3 \frac{\partial\sigma_{jl}}{\partial x_l}\quad\text{for }1\leq j\leq3\,.
\end{equation}
By taking time derivative in \eqref{E:3.25} and using \eqref{E:3.24} one gets
\begin{equation}\label{E:3.34}
\frac{\partial\sigma_{jl}}{\partial t}=\frac{1}{2\rho}\sum_{m,n=1}^3 c_{jl,mn}
\left(\frac{\partial P_m}{\partial x_n}+\frac{\partial P_n}{\partial x_m}\right)\,.
\end{equation}
By direct verification, \eqref{E:3.33}, \eqref{E:3.34} can be rewritten the form \eqref{E:3.12} with
the following identifications (recall that $\sigma_{jl}=\sigma_{lj}$): $\bsPsi(x,t)$ is the $9\times1$ vector
\begin{equation}\label{E:3.35}
\bsPsi(x,t)=\left[\begin{matrix}
\boldsymbol\Sigma(x,t) \\ \bsP(x,t)
\end{matrix}\right]\,,
\end{equation}
$\bsE(x)$ is the $9\times9$ strictly positive real matrix
\begin{equation}\label{E:3.36}
\bsE(x)=
\left[\begin{matrix}
\rho(x) \bsC(x)^{-1} & 0\\
0 & \bf 1_3
\end{matrix}\right]\,,
\end{equation}
and 
\begin{equation}\label{E:3.37}
\bsAj(x)=\bsAj=
-\left[\begin{matrix}
0 & \bsaj\\
(\bsaj)^T & 0
\end{matrix}\right]
\end{equation}
where $\bsaj$ are the $6\times3$ matrices given by
\begin{equation}\label{E:3.38}
\boldsymbol{a^1}=
\left[\begin{matrix}
1&0&0\\
0&0&0\\
0&0&0\\
0&1&0\\
0&0&0\\
0&0&1
\end{matrix}\right]\,,\quad
\boldsymbol{a^2}=
\left[\begin{matrix}
0&0&0\\
0&1&0\\
0&0&0\\
1&0&0\\
0&0&1\\
0&0&0
\end{matrix}\right]\,,\quad
\boldsymbol{a^3}=
\left[\begin{matrix}
0&0&0\\
0&0&0\\
0&0&1\\
0&0&0\\
0&1&0\\
1&0&0
\end{matrix}\right]\,.
\end{equation}
As in the case of Maxwell equations one can write $\bsM$ in terms of $\bsC$ and $\bsaj$:
\begin{equation}\label{E:vel}
M_{jl}=\frac{2}{\rho}\text{Tr}\big(\bsC^{1/2}\bsaj(\bsC^{1/2}\bsal)^T\big).
\end{equation}

In  this general case the matrix $\bsM(x)$ is a complicated function of $\rho(x)$ and $\bsC(x)$, 
and so the physical heuristics behind Theorem~\ref{T:1} remains unclear. However, as in the case
of Maxwell's equations above, the picture becomes clear in the isotropic case.
More precisely, on one hand the stiffness matrix $\bsC$ takes the simple form
\begin{equation}\label{E:3.39}
\bsC=
\left[\begin{matrix}
K+\frac{4\mu}{3} & K-\frac{2\mu}{3} & K-\frac{2\mu}{3} & & & \\
K-\frac{2\mu}{3} & K+\frac{4\mu}{3} & K-\frac{2\mu}{3} & &0& \\
K-\frac{2\mu}{3} & K-\frac{2\mu}{3} & K+\frac{4\mu}{3} & & & \\
 & & & \mu & & \\
 & 0 & & & \mu & \\
 & & & & & \mu \\  
\end{matrix}\right]
\end{equation}
where $K,\mu>0$ are the so-called bulk and shear moduli, respectively. On the other hand, 
there are two characteristic velocities of propagation, namely compressional and
shear velocities, $v_P$ and $v_S$ respectively, which can be expressed
in terms of $K,\mu,$ and $\rho$:
\begin{equation}\label{E:3.40}
v_P^2=\frac{1}{\rho}\left(K+\frac{4\mu}{3}\right)\quad\text{and}\quad v_S^2=\frac{\mu}{\rho}\,.
\end{equation}

Now a straightforward, but somewhat tedious, computation gives:
\begin{equation}\label{E:3.41}
M_{jl}
=\frac{2}{\rho}\left(K+\frac{10\mu}{3}\right)\,\delta_{jl}.
\end{equation}
Combining \eqref{E:3.40} and \eqref{E:3.41} one obtains
\begin{equation}\label{E:3.42}
\bsM(x)=2\big(v_P(x)^2+2v_S(x)^2\big)\bf {1}\,,
\end{equation}
i.e. $M$ can be expressed solely in terms of propagation velocities of the system.
Hence, as in the case of Maxwell's equations, the essential self-adjointness of $\ID_0$
(i.e. the fact that the total energy is conserved in time) is ensured by sufficient
decay of both $v_P$ and $v
_S$ as $x\to\partial\Om$. In particular,
Corollary~\ref{C:4} holds true with $c(x)$ replaced by $\max\big\{v_P(x),v_S(x)\big\}$.
\end{example}



\begin{thebibliography}{10}

\bibitem{Ag} S.~Agmon, \textit{Lectures on exponential decay of solutions of second-order
elliptic equations: bounds on eigenfunctions of $N$-body Schr\"odinger operators}, Mathematical Notes, {\bf 29},
Princeton University Press, 1982.

\bibitem{AG} N.~I.~Akhiezer, I.~M.~Glazman, \textit{Theory of linear operators in Hilbert spaces. Vol II.}
Translated from the third Russian edition by E.~R.~Dawson. Translation edited by W.~N.~Everitt.
Monographs and Studies in Mathematics, 10. \textit{Pitman (Advanced Publishing Program), Boston, Mass.-London}, 1981.

\bibitem{BFSB} R.~D.~Benguria, S.~Fournais, E.~Stockmeyer, H.~van Den Bosch, Self-adjointness of 
two-dimensional Dirac operators on domains, \textit{Ann. Henri Poincar\'e}, {\bf 18} (2017), 1371-1383.

\bibitem{BMS} M.~Braverman, O.~Milatovich, M.~Shubin, Essential self-adjointness of Schr\"odinger-type operators
on manifolds, \textit{Russian Math. Surveys} {\bf 57} (2002), 641--692.

\bibitem{Br} A. ~G. ~Brusentsev, Self-adjointness of elliptic differential operators in $L^2(G)$ and correcting potentials, \textit{ Trans. Moscow Math.. Soc.}
{\bf 65} (2004, 31-61.

\bibitem{Ch1} P.R.~Chernoff, Essential self-adjointness of powers of generators of hyperbolic equations,
\textit{J. Functional Analysis} {\bf 12} (1973), 401--414.

\bibitem{Ch2} P.R.~Chernoff, Schr\"odinger and Dirac operators with singular potentials and hyperbolic equations, \textit{Pacific J. Math.} {\bf 72} (1977), 361--382.

\bibitem{EL} M.~J.~ Esteban, M.~Loos, Self-adjointness for Dirac operators via Hardy-Dirac inequalities,
\textit{J. Math. Phys.} {\bf 48} (2007), 112107.

\bibitem{Fa} H.~O~Fattorini, Weak and strong extensions of first-order differential operators in ${\bf R}^n$, 
\textit{Journal of Differential Equations} {\bf 34} (1979), 353-360.

\bibitem{FR} F.~ Finster, C.~R\"oken, Self-adjointness of the Dirac hamiltonian for a class of non-uniformly elliptic 
boundary value problems. \textit{Annals of Mathematical Sciences and Applications} {\bf 1} (2016), 301-320.

\bibitem{Fr} K.~Friedrichs, The identity of weak and strong extensions of differential operators, \textit{Trans. Amer. Math. Soc.} {\bf 55} 
(1944), 132--151.

\bibitem{Ho} L.~H\"ormander, Weak and strong extensions of differential operators, \textit{Comm. Pure Appl. Math} {\bf 14}
(1961), 371--379.

\bibitem{J} F.~John, \textit{Partial Differential Equations}, Fourth Edition, Springer-Verlag, New York 1982. 

\bibitem{Jo} J.~Jost, \textit{Riemannian geometry and geometric analysis.} Springer 2002.

\bibitem{KSWW} H.~Kalf, U.-V.~Schmincke, J.~Walter, R.~W\"ust, On the spectral theory of Schr\"odinger and Dirac 
operators with strongly singular potentials, \textit{Spectral theory and differential equations (Proc. Sympos., Dundee, 1974; dedicated to Konrad J\"orgens)}, pp. 182--226. Lecture Notes in Math., Vol. 448, Springer, Berlin, 1975.

\bibitem{LM} M.~Lesch, M.~Malamud, On the deficiency indices of symmetric Hamiltonian
systes, \textit{J. Differential Equations} {\bf 189} (2003), 556--615.

\bibitem{MT} O. Milatovic, F. Truc, Self-adjoint extensions of differential operators on Riemannian manifolds,
\textit{Ann. Global Anal. Geom.} {\bf 49} (2016), no.1, 87--103.

\bibitem{NN2} G.~Nenciu, I.~Nenciu, Drift-diffusion equations on domains in $\IR^d$:
essential self-adjointness and stochastic completeness, \textit{J. Functional Analysis} {\bf 273}
(2017), 2619--2654.

\bibitem{P} A.D.~Pierce, Basic Linear Acoustics, in {\em Springer Handbook
 of Acoustics}, T.D. Rossing Ed. Springer 2007.
 
 \bibitem{PRS} D.~Prandi, L.~Rizzi, M,~Seri, Quantum confinement on non-complete Riemannian manifolds. 
{\em Journal of Spectral Theory} (in press).
  
\bibitem{RS} M.~Reed, B.~Simon, \textit{Methods of modern mathematical physics. II. Fourier analysis, 
self-adjointness.} Academic Press, New York-London, 1975. xv+361 pp.

\bibitem{Sh} M.~Shubin, Essential self-adjointness for semi-bounded magnetic Schr\"odinger
operators on non-compact manifolds, \textit{J. Functional Analysis} {\bf 186} (2001), 92--116.

\bibitem{Th} B. Thaller, \textit{The Dirac equation}, Springer-Verlag, Berlin 1992.

\bibitem{Wi1} C.H.~Wilcox, Wave operators and asymptotic solutions of wave propagation
problems of classical physics, \textit{Archive for Rational Mechanics and Analysis} {\bf 22} (1966),
37--78.

\bibitem{Wi2} C.H.~Wilcox, The domain of dependence inequality for symmetric hyperbolic
systems, \textit{Bull. AMS} {\bf 70} (1964), 149--154.

\end{thebibliography}
\end{document}